\documentclass[smallextended]{svjour3}     

\smartqed 

\usepackage{amssymb}
\usepackage{amsfonts} 
\usepackage{color}
\usepackage{ulem}
\usepackage{algorithmic}

\spnewtheorem{alg}{Algorithm}{\bfseries}{\rmfamily}
\usepackage{graphicx}

\usepackage{mathptmx}      
\journalname{Des. Codes Cryptogr.}
\begin{document}

\title{Computing coset leaders of binary codes\thanks{Partially supported by Spanish {AECID A/016959/08}.}
}


\author{M. Borges-Quintana\and M.A.~Borges-Trenard\and E.~Mart\'{\i}nez-Moro
}


\institute{M. Borges-Quintana \and M.A. Borges-Trenard  \at
              Dpto. de Matem\'atica, Faculty of Mathematics and Computer Science,   
              Universidad de Oriente,\\ Santiago de Cuba, Cuba.\\
              \email{mijail@csd.uo.edu.cu, mborges@csd.uo.edu.cu}          
           \and
           E. Mart\'{\i}nez-Moro \at
           Dpto. de Matem\'atica Aplicada, Universidad  de Valladolid, Castilla, Spain.\\ 
           \email{edgar@maf.uva.es}
}

\date{Received: date / Revised: date}
\maketitle
\begin{abstract} 

We present an algorithm for computing the set of all coset leaders of a binary code $\mathcal C \subset \mathbb{F}_2^n$. The method is adapted from  some of the techniques related to the computation of Gr\"obner representations associated with codes.  
The algorithm provides a Gr\"obner representation of the binary code and the set of coset leaders $\mathrm{CL}(\mathcal C)$. Its efficiency  stands of the fact that its complexity is linear on the number of elements of $\mathrm{CL}(\mathcal C)$, which  is  smaller than  exhaustive search in $\mathbb{F}_2^n$.

\end{abstract}

\section{Introduction}\label{s:int}
The error-correction problem in coding theory addresses given a received word recovering  the codeword  closest to it with respect to the Hamming  
distance. This previous  statement is the  usual formulation of the  \textit{Complete Decoding Problem} (CDP). The $t$-bounded  distance decoding ($t$-BDD) algorithms determine a codeword (if such a word exists) which is at distance less or equal to $t$ to the received word. If $t$ is the covering radius of the code then the bounded distance decoding problem is the same as CDP.
In the CDP of a linear code  of length $n$ , $\mathcal C\subset\mathbb F_q^n$ those errors that can be corrected are  just the coset leaders, which are vectors of smallest weight in the cosets  $\mathbb F_q^n/{\mathcal C}$. When there is more than one leader in a coset there is more than one choice for the error.
Therefore the following problem  naturally follows known as the \textit{coset weights problem}(CWP):
\begin{description}
 \item{\textbf{Input:}} A binary $r\times n$ matrix $H$, a vector $\mathbf s\in \mathbb F_2^r$ and a non-negative integer $t$.
\item{\textbf{Problem:}} Does a binary vector $\mathbf e \in \mathbb F_2^r$ of Hamming weight at most $t$ exist such that $H\mathbf e=s$?
\end{description}
Recall that $H$ can be seen as the parity check matrix of a binary code and $\mathbf s$ the syndrome of a received word, thus the knowledge of $\mathbf e$ would solve the $t$-BDD problem. Unfortunatelly the hope of finding an efficient $t$-BDD algorithm is very bleak since it was proven that the CWP is NP-complete \cite{Berle}. Thus the computation of all the coset leaders is also NP-complete. The study of the set of coset leaders is also related to the study of the set of minimal codewords, which have been used in the Maximum Likelihood Decoding Analysis \cite{Barg} and which are also related to the minimal access structure of secret sharing schemes \cite{massey}. Furthermore, the computation of all coset leaders of a code allows to know more about its internal structure \cite{bbm1}.

All problems mentioned before are considered to be hard computational problems (see for example \cite{Barg,Berle}) even if preprocesing is allowed \cite{BN}. However, taking into account the nature of the problem, to develop an algorithm for computing the set of all coset leaders of a binary code, in the vector space of $2^n$ vectors, by generating a number of vectors close to the cardinality of this set may be quite efficient. This is our purpose extending some results on Gr\"obner representations for binary codes. For previous results and applications of Gr\"obner representations for linear codes see \cite{bbm2,bbm1,bbfm2}, for a summary of the whole material we refer the reader to \cite{bbm3}. We extend some settings of previous work in order  to obtain the set of all coset leaders  keeping record of the additive structure in the cosets. The presentation of the paper is done in a ``{Gr\"obner bases}''-free context.

The outline of the paper is as follows. In Section~\ref{s:pre} we review some of the standard facts on binary linear codes  and  their Gr\"obner representation. 
Section~\ref{s:cl} gives a concise presentation of the results in this paper. Definition~\ref{d:list} corresponds to the construction of $\mathrm{List}$, the main object that is used in the algorithm proposed, while Theorem~\ref{t:clInList} guarantees that all coset leaders will belong to $\mathrm{List}$. In this section it is presented the algorithm {\rm CLBC} for computing the set of all coset leaders. At the end of the section we show an example to illustrate a computational approach applied to a binary linear code  with 64 cosets and 118 coset leaders.
In Section~\ref{s:comp} we discuss some complexity issues. Finally some conclusion are given which include further research.

\section{Preliminaries}\label{s:pre}
Let $\mathbb{F}_2$ be  the finite field with $2$   
elements and $\mathbb{F}_2^n$ be the $\mathbb{F}_2$-vector space of dimension $n$.   
We will call the vectors in  $\mathbb{F}_2^n$ \textit{words}. A linear code $\mathcal C$ of dimension $k$ and length $n$  is the image of an injective linear mapping $L: \mathbb{F}_2^k\to \mathbb{F}_2^n$,  where $k\leq n$, i.e. $\mathcal{C}=L(\mathbb{F}_2^k)$. From now on, we will use the term code to mean binary linear code. The elements in $\mathcal C$ are called {\sl codewords}. For a word $\mathbf y\in\mathbb F_2 ^n$  the \textit{support} of $\mathbf y$ is $\mathrm{supp}(\mathbf y)=\{ i\in \{1,\ldots, n\}\mid \mathbf y_i\neq 0 \}$ and \textit{the Hamming  weight of} $\mathbf y$  is given by $\mathrm{weight}(\mathbf y)$ the cardinal of $\mathrm{supp}(\mathbf y)$. The Hamming distance between two words $\mathbf c_1, \mathbf c_2$ is $\mathrm{d}(\mathbf c_1,\mathbf c_2)=\mathrm{weight}(\mathbf c_1-\mathbf c_2)$ and the minimum distance $d$  
of a code is the minimum weight among all the non-zero codewords.  It is well known that CDP has a unique solution for those vectors in $B({\mathcal C},t)=\{\mathbf y\in  
\mathbb{F}_2^n\mid \exists\, \mathbf c\in\mathcal{C}\hbox{ s.t. } d(\mathbf c,\mathbf y)\leq  \left[\frac{d-1}{2}\right] \} $ where $\left[\cdot\right]$ is the  
greatest integer function.

\begin{definition}
 The words of minimum Hamming weight in the cosets of $\mathbb{F}_2^n/\mathcal C$ are called \textit{coset leaders}.
\end{definition}

Cosets corresponding to  $B({\mathcal C},t)$ have a unique leader however in general outside $B({\mathcal C},t)$ there may be also cosets with a unique leader, i.e. those cosets with only one leader could be  more than $| B({\mathcal C},t)|$. 
Let $\mathrm{CL}(\mathcal C)$ denote the set of coset leaders of the code $\mathcal C$ and 
$\mathrm{CL}(\mathbf y)$ the subset of coset leaders corresponding to the coset $\mathcal C+\mathbf y$. Let $\mathbf e_i$ $i=1,\ldots, n$ be  the $i$-th vector of the canonical basis $\mathrm X=\{\mathbf e_1, \ldots,\,\mathbf e_n\}\subset \mathbb{F}_2^n$  and $\mathbf 0$ the zero vector of $\mathbb{F}_2^n$.
The following theorem gives some relations between cosets  \cite[Corollary~11.7.7]{Pless}.  
\begin{theorem}\label{t:cl}  
Let $\mathbf w \in \mathrm{CL}(\mathcal C)$, such that $\mathbf w = \mathbf w_1 + \mathbf e_i$ for some $\mathbf w_1 \in \mathbb{F}_2^n$ and $i \in \{1,\ldots, n\}$. Then, $\mathbf w_1 \in \mathrm{CL}(\mathbf w_1)$.  
\end{theorem}  

\begin{definition}
 A  \textit{Gr\"obner representation} of $\mathbb F_2^n/\mathcal C$ \cite{bbm2,bbm3} is a pair $\mathrm N,\phi$ where ${\mathrm N}$ is a transversal of $\mathbb F_2^n/\mathcal C$ such that $\mathbf 0\in N$ and for each $\mathbf n\in {\mathrm N} \setminus\{\mathbf 0\}$ there exists $\mathbf e_i$, $i \in 1,\ldots,n$, such that $\mathbf n=\mathbf n^\prime+\mathbf e_i$ and $\mathbf n^\prime\in {\mathrm N}$, and $\phi: {\mathrm N}\times\{\mathbf e_i\}_{i=1}^n\rightarrow  {\mathrm N}
$ be the function that maps each pair $(\mathbf n,\mathbf e_i)$ to the element of  ${\mathrm N}$ which belongs to the coset of $\mathbf n+\mathbf e_i$.  
\end{definition}

\section{Computing the set of coset leaders}\label{s:cl}
A key point of the algorithms for computing Gr\"obner representations is the construction of an object we called $\mathrm{List}$ which is an ordered set of elements of  $\mathbb{F}_2^n$ w.r.t. a linear order $\prec$ defined as follows: $\mathbf{w}\prec\mathbf{v}$ if $\mathrm{weight}(\mathbf{w})<\mathrm{weight}(\mathbf{v})$ or $\mathrm{weight}(\mathbf{w})=\mathrm{weight}(\mathbf{v})$ and $\mathrm{weight}(\mathbf{w})\prec_{1}\mathrm{weight}(\mathbf{v})$, where $\prec_{1}$ is any admissible order on $\mathbb{F}_2^n$ in the sense given in \cite[p. 167]{Weisp}. We will call such kind of orders as weight compatible orderings.

\begin{definition}[Construction of $\mathrm{List}$].\label{d:list} Let $\mathrm{List}$ be the ordered structure given by the following axioms 
\begin{enumerate}
\item $\mathbf 0 \in \mathrm{List}$.
\item If $\mathbf v \in \mathrm{List}$, let $\mathrm{N}(\mathbf v)=\min_{\prec}\{\mathbf{w} \mid \mathbf w \in \mathrm{List} 
\cap (\mathcal C+\mathbf v)\}$.
\item If $\mathbf v \in \mathrm{List}$ is such that $\mathrm{weight}(\mathbf v)=\mathrm{weight}(\mathrm N(\mathbf v))$, then $\{\mathbf v + \mathbf e_i\,: i\notin \mathrm{supp}(\mathbf v), \,i\in \{1,\ldots, n\}\} \subset \mathrm{List}$.
\end{enumerate}
\end{definition}

\begin{remark}\label{r:NL}
The set $\mathrm{N}$ is the subset of $\mathrm{List}$ such that $\mathrm{N}(\mathbf v)$ is the least element in $\mathrm{List} \cap (\mathcal C +\mathbf v)$, i.e. $\mathrm{N}(\mathbf v) \in \mathcal C +\mathbf v$ and $\mathrm{N}(\mathbf v)=\min_{\prec}\{\mathbf{w} \mid \mathbf w \in \mathrm{List} 
\cap (\mathcal C+\mathbf v)\}$.\end{remark}

\begin{remark}\label{r:dlist}
We start $\mathrm{List}$ with $\mathbf 0$ and we will see that, whenever condition \textit{3.} holds for $\mathbf v$, the vector $\mathbf v$ will be a coset 
leader of $\mathcal C+\mathbf v$. Then, we insert $\mathbf v + \mathbf e_i$ to $\mathrm{List}$, for $i=1,\ldots ,n$ and $i\notin 
\mathrm{supp}(\mathbf v)$ (see Theorem \ref{t:cl}). It is not necessary to introduce $\mathbf v + \mathbf e_i$, for $i\in\mathrm{supp}(\mathbf v)$, because in this case $\mathbf v + \mathbf e_i \prec\mathbf v$ and then $\mathbf v + \mathbf e_i$ has been already considered in $\mathrm{List}$ since it is an ordered structure.\end{remark}

Next theorem states that $\mathrm{List}$ in Definition~\ref{d:list}  includes the set of coset leaders.
\begin{theorem}\label{t:clInList}
Let $\mathbf w \in \mathrm{CL}(\mathcal C)$, then $\mathbf w \in \mathrm{List}$.
\end{theorem}
\begin{proof}
We will proceed by Noetherian Induction on the words of $\mathbb{F}_2^n$ with the ordering $\prec$, { since $\prec$ has the property of being a well-$\;$founded ordering or  Noetherian ordering (i.e., any descending chain of words is finite) \cite[Theorem 4.62, p. 168]{Weisp}}.  
The word ${\mathbf 0} \in \mathrm{CL}({\mathcal C})$ and by definition it belongs to $\mathrm{List}$, let $\mathbf w \in \mathrm{CL}({\mathcal C})\setminus \{\mathbf 0\}$.
Assume the property valid for any word less than $\mathbf w$ with respect to $\prec$, i.e. $\mathbf u \in \mathrm{List}$ provided $\mathbf u \in \mathrm{CL}(\mathcal C)$ and $\mathbf u \prec \mathbf w$. Taking $i \in \mathrm{supp}(\mathbf w)$, we write ${\mathbf w}={\mathbf u}+{\mathbf e}_i$ where $\mathbf u \in \mathbb{F}_2^n$, then by Theorem~\ref{t:cl}, $\mathbf u \in \mathrm{CL}({\mathcal C})$. In addition, $\mathrm{supp}(\mathbf u)=\mathrm{supp}(\mathbf w)-1$, thus $\mathbf u \prec \mathbf w$.  Therefore, by applying the induction principle we have $\mathbf u \in \mathrm{List}$. If $\mathbf u$ is a coset leader belonging to $\mathrm{List}$ it is clear that $\mathrm{weight}(\mathbf u)=\mathrm{weight}(\mathrm{N}(\mathbf u))$. As a consequence, by \textit{3.} in Definition~\ref{d:list}, $\mathbf w = \mathbf u + \mathbf e_i\in \mathrm{List}$. \qed
\end{proof}

\begin{alg}[CLBC]\label{a:cl} $ $
\begin{description}
 \item[\textbf{Input:}] A parity check matrix of a binary code $\mathcal C$.
\item[\textbf{Output:}] $\mathrm{CL}(\mathcal C),\phi$: The set of all coset leaders and the function Matphi.
\end{description}
\begin{algorithmic}[1]
\STATE $\mathrm{List}\leftarrow [\mathbf 0]$, $\mathrm{N}\leftarrow \emptyset$, $r\leftarrow 0$, $\mathrm{CL}(\mathcal C)\leftarrow []$ 
\WHILE{$\mathrm{List}\neq \emptyset$}
\STATE $\tau \leftarrow \mathrm{NextTerm}[\mathrm{List}]$, $\mathbf s\leftarrow \tau H$ 
\STATE $j\leftarrow \mathrm{Member}[\mathbf s,\{\mathbf s_1,\dots,\mathbf s_r\}]$
\IF{$j\neq\mathrm{false}$}
\STATE {\bf for} $k$ such that $\tau=\tau^\prime + \mathbf e_k$ with $\tau^\prime\in \mathrm{N}$ {\bf do} $\phi(\tau^\prime,\mathbf e_k)\leftarrow \tau_j$
\IF{$\mathrm{weight}(\tau)=\mathrm{weight}(\tau_j)$}
\STATE { 
$\mathrm{CL}(\mathcal C)[\tau_j]\leftarrow \mathrm{CL}(\mathcal C)[\tau_j]\cup \{\tau\}$}
\STATE $\mathrm{List}\leftarrow \mathrm{InsertNext}[\tau,\mathrm{List}]$
\ENDIF
\ELSE
\STATE $r\leftarrow r+1$, $\mathbf s_r\leftarrow\mathbf s$, $\tau_r\leftarrow \tau$, $\mathrm{N}\leftarrow \mathrm{N}\cup\{\tau_r\}$
\STATE {
$\mathrm{CL}(\mathcal C)[\tau_r]\leftarrow \mathrm{CL}(\mathcal C)[\tau_r]\cup \{\tau\}$}
\STATE $\mathrm{List}\leftarrow \mathrm{InsertNext}[\tau_r,\mathrm{List}]$ 
\FOR{$k$ such that $\tau_r=\tau^\prime + \mathbf e_k$ with $\tau^\prime\in \mathrm{N}$}
\STATE $\phi(\tau^\prime,\mathbf e_k)\leftarrow \tau_r$ 
\STATE $\phi(\tau_r,\mathbf e_k) \leftarrow \tau^\prime$ 
\ENDFOR
\ENDIF
\ENDWHILE
{
\RETURN $\mathrm{CL}(\mathcal C),\, \phi$} 
\end{algorithmic}
\end{alg}

Where 
\begin{enumerate}
\item $\mathrm{InsertNext}[\tau,\mathrm{List}]$ Inserts all the sums $\tau + \mathbf e_k$ in $\mathrm{List}$, where $k \notin \mathrm{supp}(\tau)$, and keeps $\mathrm{List}$ in increasing order w.r.t. the ordering $\prec$.
\item $\mathrm{NextTerm}[\mathrm{List}]$ returns the first element from $\mathrm{List}$ and deletes it from this set. 
\item $\mathrm{Member}[obj,G]$ returns the position $j$ of $obj$ in $G$ if $obj\in G$ and false otherwise.
\end{enumerate}

\begin{theorem}\label{t:correcto}
{\rm CLBC} computes the set of coset leaders of a given binary code and its corresponding {\rm Matphi}.
\end{theorem}
\begin{proof}
Note that when an element $\tau$ is taken out from $\mathrm{List}$ by $\mathrm{NextTerm}$ the elements to be inserted in $\mathrm{List}$ by $\mathrm{InsertNext}$ are of the form $\tau + \mathbf e_k$, where $k \notin \mathrm{supp}(\tau)$. As a consequence, $\tau + \mathbf e_k \succ \tau$. Then all elements generated by {\rm CLBC} in $\mathrm{List}$, after $\tau$ is taken out, are greater than $\tau$. Therefore, when $\tau$ is the first element in $\mathrm{List}$ in Step~3, all elements of $\mathrm{List}$ that shall be analyzed by {\rm CLBC} are greater than $\tau$.

Let us prove that the procedure generates $\mathrm{List}$ according to Definition~\ref{d:list}. First, by Step~1, $\mathbf 0 \in \mathrm{List}$. Let $\tau = \mathrm{NextTerm}[\mathrm{List}]$ in Step 3, in this step the syndrome $\mathbf s$ of $\tau$ is computed. Thus we have two cases regarding the result of Step~4, namely

\begin{enumerate}
 \item Assume $j = \hbox{``false''}$, thus Steps 5 and 11 guaranty us to find $\mathrm{N}(\tau)$ as the least element in $\mathrm{List}$ having the same syndrome as $\tau$. When we are in this case ($\tau = \mathrm{N}(\tau)$), in Step~14 it is performed (ii) of Definition~\ref{d:list}.
\item On the other hand, assume $j \neq \hbox{``false''}$ (an element $\mathrm{N}(\tau)=\tau_j$ have been already computed), if we have $\mathrm{weight}(\tau)=\mathrm{weight}(\tau_j)$, in Step~9 it is performed (ii) of Definition~\ref{d:list}.
\end{enumerate}
Therefore {\rm CLBC} constructs the object $\mathrm{List}$ following Definition~\ref{d:list}.
By Theorem~\ref{t:clInList}, the set of coset leaders is a subset of $\mathrm{List}$. Then Steps 8 and 13 assure the computation of this set.
The procedure computes the {\rm Matphi} structure as a direct consequence of Steps~6, 16, 17. The reason for including Step~17 is that in this case $$\tau_r + \mathbf e_k=\tau^\prime\prec \tau_r$$ and those elements $\tau_r+\mathbf e_k$, where $k\in \mathrm{supp}(\tau_r)$, are not inserted  in $\mathrm{List}$ since $\tau^\prime$ has been already considered when $\tau_r$ is computed as a new element of $\mathrm{N}$. In addition, { since $\prec$ is admissible}, all subwords of $\tau_r$ are also least elements of their cosets according to $\prec$, so $\tau^\prime=\mathrm{N}(\tau^\prime)$ (note $\tau^\prime$ is a coset leader and by Theorem~\ref{t:clInList} it belongs to $\mathrm{List}$).
                               
We have proved that {\rm CLBC} guarantees the outputs we are claiming. Termination is a consequence of the fact that the cardinal of the set of elements belonging to $\mathrm{List}$ is least than $n|\mathrm{CL}({\mathcal C})|$. Then, to a certain extent (when the set of coset leaders has been computed) no more elements are inserted in $\mathrm{List}$ in Steps~9 and 14. Therefore, the list get empty, and by Step~2, {\rm CLBC} terminates.\qed
\end{proof}
\begin{remark}\label{r:alg} $\,$
\begin{enumerate}
\item We introduce in this paper steps 7, 8 and 9. In our predecessor algorithms where it was not necessary to introduce them \cite{bbm2,bbm1,bbm3}.  These steps, as well as a modified definition of list,  are forced by the fact that  all coset leaders are incorporated.
\item In Steps from 14 to 17, the pairs $(\tau_r, \mathbf e_k)$, where $k \in \mathrm{supp}(\tau_r)$, are not necessary to compute the coset leaders but for computing the structure {\rm Matphi} and could be removed if we are only interested in the coset leaders.
\end{enumerate}
\end{remark}
\begin{example}\label{s:exm}  
Consider the $[10,4,1]$ code  with parity check matrix $H$

$$H=\left( \begin{array}{cccccccccc}
1& 0& 0& 0& 1& 0& 0& 0& 0 & 0\\
1& 0& 1& 1& 0& 1& 0& 0& 0& 0\\
1& 1& 0& 1& 0& 0& 1& 0& 0& 0\\
1& 1& 1& 0& 0& 0& 0& 1& 0& 0\\
1& 1& 1& 1& 0& 0& 0& 0& 1& 0\\
1& 1& 1& 1& 0& 0& 0& 0& 0& 1
\end{array}
\right)  .$$

We use GAP 4.12 \cite{GAP} and the GAP's package GUAVA 3.10 for Coding Theory. We have built in this framework a collection of programs we call GBLA\_LC ``Gr\"obner Bases by Linear Algebra and Linear Codes" \cite{GBLA}. In particular, we have run the function ``{\rm CLBC}" of GBLA\_LC (Coset Leaders of Binary Codes), it gives a list of three objects as an output, the first one is the set of coset leaders, the second one the function Matphi, and the third one the error correcting capability of the code. The complete set of coset leaders is the list below with 64 components and each component corresponds to a coset with its cosets leaders, the set $\mathrm{N}$ is composed by the first elements of each component. We have indicated with arrows some places in the list below, which we are going to use during the example. The elements in the list $\mathrm{CL}({\mathcal C})$ are

$$\begin{array}{l}
[\,[ 1 ], [ e_1 ], [ e_2 ], [ e_3 ], [ e_4 ], [ e_5 ], [ e_6 ], [ e_7 ],[e_8 ], [ e_9 ], [ e_{10} ],\\

[ e_1+e_2,\rightarrow e_5 + e_6\leftarrow],[ e_1+e_3, e_5+e_7],[ e_1+e_4, e_5+e_8 ],\\

[e_1+e_5, e_2+e_6, e_3+e_7, e_4+e_8 ],[ e_1+e_6, e_2+e_5 ],[e_1+e_7, e_3+e_5 ],\\ 

[ e_1+e_8, \rightarrow e_4+e_5\leftarrow],[ e_1+e_9 ],[ e_1+e_{10} ],[ e_2+e_3, e_6+e_7 ],\\

[ e_2+e_4, e_6+e_8 ],[ e_2+e_7, e_3+e_6 ],[ e_2+e_8, \rightarrow e_4+e_6\leftarrow],[ e_2+e_9 ], [ e_2+e_{10} ],\\

[ e_3+e_4, e_7+e_8 ],  [ e_3+e_8, e_4+e_7 ], [ e_3+e_9 ], [ e_3+e_{10} ],\\

[ e_4+e_9 ], [ e_4+e_{10} ], [ e_5+e_9 ], [ e_5+e_{10} ], [ e_6+e_9 ],[ e_6+e_{10} ],\\

[ e_7+e_9 ], [ e_7+e_{10} ], [ e_8+e_9 ], [ e_8+e_{10} ],[ e_9+e_{10} ],\\

[ e_1+e_2+e_3, e_1+e_6+e_7, e_2+e_5+e_7, e_3+e_5+e_6 ],\\

\rightarrow[ e_1+e_2+e_4, e_1+e_6+e_8, e_2+e_5+e_8, e_4+e_5+e_6 ]\leftarrow,\\

[ e_1+e_2+e_7, e_1+e_3+e_6, e_2+e_3+e_5, e_5+e_6+e_7 ],\\

[ e_1+e_2+e_8, e_1+e_4+e_6, e_2+e_4+e_5, e_5+e_6+e_8 ],\\

[ e_1+e_2+e_9, e_5+e_6+e_9 ], [ e_1+e_2+e_{10}, e_5+e_6+e_{10} ],\\

[ e_1+e_3+e_4, e_1+e_7+e_8, e_3+e_5+e_8, e_4+e_5+e_7 ],\\

[ e_1+e_3+e_8, e_1+e_4+e_7, e_3+e_4+e_5, e_5+e_7+e_8 ],\\

[ e_1+e_3+e_9, e_5+e_7+e_9 ], [ e_1+e_3+e_{10}, e_5+e_7+e_{10} ],[ e_1+e_4+e_9, e_5+e_8+e_9 ],\\

[ e_1+e_4+e_{10}, e_5+e_8+e_{10} ],[ e_1+e_5+e_9, e_2+e_6+e_9, e_3+e_7+e_9, e_4+e_8+e_9 ],\\

[ e_1+e_5+e_{10}, e_2+e_6+e_{10}, e_3+e_7+e_{10}, e_4+e_8+e_{10} ],[ e_1+e_6+e_9, e_2+e_5+e_9 ],\\

[ e_1+e_6+e_{10}, e_2+e_5+e_{10} ],[ e_1+e_7+e_9, e_3+e_5+e_9 ], [ e_1+e_7+e_{10}, e_3+e_5+e_{10} ],\\

[ e_1+e_8+e_9, e_4+e_5+e_9 ], [ e_1+e_8+e_{10}, e_4+e_5+e_{10} ],[ e_1+e_9+e_{10} ],\\

[ e_2+e_3+e_8, e_2+e_4+e_7, e_3+e_4+e_6, e_6+e_7+e_8 ],[ e_5+e_9+e_{10}]\,]
\end{array}$$

Note that $\mathbf y= \mathbf e_4+\mathbf e_5+ \mathbf e_6 $ in $\mathrm{CL}({\mathcal C})_{46}=[ e_1+e_2+e_4, e_1+e_6+e_8, e_2+e_5+e_8, e_4+e_5+e_6 ]$ (pointed by arrows) is a coset leader such that none subword of $\mathbf y$ is in $\mathrm{N}$ (see the previous elements pointed by arrows). This shows the importance of considering all coset leaders in \textit{3.} of Definition~\ref{d:list} and not only the coset leaders belonging to $\mathrm{N}$ like in previous works.

The algorithm could be adapted without incrementing the complexity to get more information like the Covering radius and the Newton radius. In this case by analyzing the last element $\mathrm{CL}({\mathcal C})_{64}=[ e_5+e_9+e_{10}]$ we have a coset of highest weight which also contains only one leader; therefore, the Covering radius and the Newton radius are equal to 3. Moreover, for computing these parameters the algorithm do not need to run until the very end.

As we state before, the set $\mathrm{N}$ would be enough to compute the Weight Distribution of the Coset Leaders $\mathrm{WDCL}=(\alpha_0,\ldots,\alpha_n)$ where $\alpha_i$ is the number of cosets with coset leaders of weight $i$, $i=1,\ldots , n$  of the code. We can provide also an object which gives more information about the structure of the code, that would be the numbers of coset leaders in each coset ($\mathrm{\#(CL)}$).$$\mathrm{WDCL}=[1,10,30,23,0,0,0,0,0,0,0].$$
$$\begin{array}{l}
\mathrm{\#(CL)}=[ 1, 1, 1, 1, 1, 1, 1, 1, 1, 1, 2, 1, 1, 1, 2, 2, 4, 1, 1, 1, 1, 2, 2, 2, 2, 2, 2, 4, 2, 2, 2, 4, 4,\\
\hspace{1.4cm}1, 1, 1, 4, 4, 4, 2, 2, 2, 2, 2, 4, 4, 1, 1, 1, 2, 2, 2, 2, 4, 1, 1, 1, 2, 2, 2, 1, 1, 1, 1 ].
\end{array}$$
It is also interesting to note that there are more cosets with one leader (19) out the cosets corresponding to $B({\mathcal C},t)$ than cosets corresponding to $B({\mathcal C},t)$ $$|\{\mathrm{CL}(\mathbf y)\mid\mathbf y\in B({\mathcal C},t)\}|=11.$$ Therefore, there are $30$ of the $64$ cosets where the \rm{CDP} has a unique solution. 
\end{example}

\section{Complexity Analysis}\label{s:comp}
For a detailed complexity analysis and some useful considerations from the computational point of view we refer the reader to \cite{bbm1}. Section~6 of that paper is devoted to discuss in details about the computational complexity and space complexity of the setting for computing Gr\"obner basis representation for binary codes and general linear codes. The method is also compared with other existent methods for similar purposes. 
In the case of this paper, the difference is that we work  out the set of all coset leaders and not only a set of canonical forms. Next theorem  shows an upper bound for the number of iterations that will perform {\rm CLBC}.

\begin{theorem}\label{t:comp}
{\rm CLBC} computes the set of coset leaders of a given binary code $\mathcal C$ of length $n$ after at most $\mathbf n|\mathrm{CL}(\mathcal C)|$ iterations.
\end{theorem}
\begin{proof}
Note that the number of iterations is exactly the size of $\mathrm{List}$. In the proof of Theorem~\ref{t:correcto} was shown that the algorithm follows this definition to construct the object $\mathrm{List}$. It is clear that the size of $\mathrm{List}$ is bounded by $n|\mathrm{CL}(\mathcal C)|$, note that we can write $\mathrm{List}$, as a set as follows

$$\mathrm{List}=\{\mathbf w + \mathbf e_i\mid \mathbf w\in\mathrm{CL}(\mathcal C)\mbox{ and }\in \left\lbrace 1,\ldots ,n\right\rbrace \}.$$
\qed
\end{proof}
\begin{remark}
\begin{enumerate}
\item By the proof above we require a memory space of $O(n|\mathrm{CL}(\mathcal C)|)$. We assume that for computing the set of coset 
leaders it is required at least $O(|\mathrm{CL}(\mathcal C)|)$; therefore, {\rm CLBC} 
is near the optimal case of memory requirements.
\item {\rm CLBC} generates at most $n|\mathrm{CL}(\mathcal C)|$ words from $\mathbb{F}_2^n$ to compute the set of coset leaders. An algorithm for computing this set needs to generate from $\mathbb{F}_2^n$ at least a subset formed by all coset leaders, i.e. $|\mathrm{CL}(\mathcal C)|$; therefore, the algorithm is near the optimal case of computational complexity.
\end{enumerate}
\end{remark}

\section{Conclusion}\label{s:con}
The Algorithm CLBC is formulated in this paper, which turns out to be quite efficient for computing all coset leaders  of a binary code from memory requirements and computational complexity view.  Although, as it is expected, the complexity of the algorithm is exponential (in the number of check positions).

The difference of {\rm CLBC} with its predecessors rely also in the computation of all coset leaders instead of a set of representative leaders for the cosets, however, it supports the computation of the function {\rm Matphi}. We remark that the computation of {\rm Matphi} is not necessary at all for the main goal of computing the coset leaders. We have kept this resource in the algorithm because this structure provides some computational advantages (see \cite{bbm2,bbm1,bbm3}) and, although the algorithm will be clearly faster without computing {\rm Matphi}, the nature of the computational complexity and space complexity will remain the same.

Unfortunately, a generalization to non-binary linear codes is not trivial from this work. The main reason seems to be that the solution based on Hamming weight compatible orderings will not continue being possible; the \textit{error vector ordering} we have used in the general approach \cite{bbm1} is not a total ordering, although it allowed to set up the computational environment in order to compute Gr\"obner representations for linear codes.

\end{document}